\documentclass[12pt]{article}

\usepackage{graphicx, subfigure}
\usepackage[top=1in,bottom=1in,left=1in,right=1in]{geometry}

\usepackage{amsfonts, amsthm, amsmath, amssymb, bbm}

\usepackage{times}
\usepackage{bm}
\usepackage{natbib}
\usepackage{moreverb,url}
\usepackage{multirow}
\usepackage{float}
\usepackage{mathtools}
\mathtoolsset{showonlyrefs}
\usepackage{color}
\definecolor{darkred}{RGB}{100,0,0}
\definecolor{darkgreen}{RGB}{0,100,0}
\definecolor{darkblue}{RGB}{0,0,150}

\usepackage{hyperref}
\hypersetup{colorlinks=true, linkcolor=darkred, citecolor=darkgreen, urlcolor=darkblue}

\usepackage[plain,noend]{algorithm2e}
\def\spacingset#1{\renewcommand{\baselinestretch}%
{#1}\small\normalsize} 

\newtheorem{thm}{Theorem}
\newtheorem{prp}{Proposition}
\newtheorem{lem}{Lemma}

\newtheorem{assump}{Assumption}

\theoremstyle{remark}

\def\beq{\begin{equation}} 
\def\eeq{\end{equation}}
\def\beqn{\begin{eqnarray*}}
\def\eeqn{\end{eqnarray*}}
\def\Bitem{\begin{itemize}\setlength{\itemsep}{.2in}}
\def\bitem{\begin{itemize}\setlength{\itemsep}{.05in}}
\def\eitem{\end{itemize}}
\def\Benum{\begin{enumerate}\setlength{\itemsep}{.2in}}
\def\benum{\begin{enumerate}\setlength{\itemsep}{.05in}}
\def\eenum{\end{enumerate}}
\def\bmult{\begin{multline*}}
\def\emult{\end{multline*}}
\def\bcenter{\begin{center}}
\def\ecenter{\end{center}}
\def\bframe{\begin{frame}}
\def\eframe{\end{frame}}

\def\cF{\mathcal{F}}

\def\cN{\mathcal{N}}

\def\cV{\mathcal{V}}

\def\bbB{\mathbb{B}}

\def\bbM{\mathbb{M}}

\newcommand{\E}{\operatorname{\mathbb{E}}}
\renewcommand{\P}{\operatorname{\mathbb{P}}}

\def\eps{\varepsilon}

\def\1{\mathbbm{1}}

\begin{document}

\title{\bf{A Robust Instrumental Variable Method Accounting for Treatment Switching in Open-Label Randomized Controlled Trials}}

\author{Andrew Ying\\
Department of Statistics and Data Science, The Wharton School,\\
University of Pennsylvania}
\date{}

\maketitle
\begin{abstract}
In a randomized controlled trial, treatment switching (also called contamination or crossover) occurs when a patient initially assigned to one treatment arm changes to another arm during the course of follow-up. Overlooking treatment switching might substantially bias the evaluation of treatment efficacy or safety. To account for treatment switching, instrumental variable (IV) methods by leveraging the initial randomized assignment as an IV serve as natural adjustment methods because they allow dependent treatment switching possibly due to underlying prognoses. However, the ``exclusion restriction'' assumption for IV methods, which requires the initial randomization to have no direct effect on the outcome, remains questionable, especially for open-label trials. We propose a robust instrumental variable estimator circumventing such a caveat. We derive large-sample properties of our proposed estimator, along with inferential tools. We conduct extensive simulations to examine the finite performance of our estimator and its associated inferential tools. An R package “ivsacim” implementing all proposed methods is freely available on R CRAN. We apply the estimator to evaluate the treatment effect of Nucleoside Reverse Transcriptase Inhibitors (NRTIs) on a safety outcome in the Optimized Treatment That Includes or Omits NRTIs trial.
\end{abstract}

\spacingset{1.5}
\section{Introduction}\label{sec:intro}
Randomized controlled trials (RCTs) remain as the golden standard for treatment effect evaluation because the randomization guarantees that subjects within different treatment groups are exchangeable. However, even under an RCT, such exchangeability might be contaminated due to treatment switching (also called contamination or crossover) after initial randomization, that is, when a patient randomized to one treatment arm changes to another arm during the course of follow-up \citep{cuzick1997adjusting, demetri2006efficacy, motzer2008efficacy, morden2011assessing, tsiatis2021estimating}. Treatment switching may happen when patients in the control arm switch to the experimental treatment which has manifested short-term effect, in the hope of improving prognosis. It can also occur due to changes in treatment guidelines, whereby treatment options might change during the course of follow-up. Treatment switching might bias the overall survival and the intent-to-treat effect. Inaccurate cost-effectiveness estimates caused by not appropriately handling treatment switching may result, and finite healthcare resources may be wasted.

Traditional methods \citep{robins1998structural, robins2000correcting, latimer2017adjusting, jimenez2017evaluating, latimer2018assessing, sullivan2020adjusting} to deal with treatment switching hinge heavily on a ``no unmeasured confounding (NUC)'' assumption that the treatment switching process is random through time, possibly conditional on measured baseline or time-varying characteristics. Unfortunately, NUC is unrealistic and tends to fail due to dependence between the switching process and the outcome of interest upon unknown factors because even under the sharp null hypothesis of no treatment effect, patients who switch treatment tend to have a different prognosis than patients who remain on their originally assigned treatment. 

To overcome unmeasured confounding, instrumental variables (IV) methods \citep{angrist1996identification, angrist2001instrumental, wooldridge2010econometric, martinussen2017instrumental} were carefully designed. An IV is a pre-treatment variable known to be associated with the treatment variable (IV relevance), to only affect the outcome through its effects on the treatment (IV exclusion restriction), and to otherwise be independent of any unmeasured confounders (IV independence). In randomized trials, the initial randomized treatment assignment can be readily leveraged as an IV because IV relevance and IV independence naturally hold. IV exclusion restriction is also reasonable, at least under double-blinded trials. For randomized trials with censored time-to-event outcomes in the presence of treatment switching, IV approaches include, for instance, the rank-preserving structural failure time model (RPSFTM) of \citet{robins1991correcting}, the structural cumulative failure time model (SCFTM) of \citep{shi2021instrumental}, and the structural cumulative survival model (SCSM) of \citet{ying2022structural}. See \citet{ying2022structural} for a comparison of these IV approaches. All aforementioned IV methods rely heavily on three IV assumptions listed above, among which the exclusion restriction might fail to hold under an open-label trial when site investigators and participants knew the treatment assignments. This is because one cannot rule out the possibility that study physicians might modify a patient's treatment course as a result of their randomized treatment assignments in a manner that may in turn directly impact the outcome in view, or patients might change their health-seeking behavior after treatment assignment. These induce an unintended direct effect of the randomized treatment on the outcome and hence violate the exclusion restriction assumption of an IV. Such acknowledgment invalidates all causal claims on the basis of the IV methods above. 

In this paper, we propose an instrumental variable estimator to handle treatment switching under a structural cumulative survival model without assuming exclusion restriction. Instead, we adapt the ``no interaction with unmeasured selection'' assumption proposed in \citet{tchetgen2021genius} to the time-varying treatment process. Our estimator allows for a time-varying treatment effect and we further develop an asymptotic framework for inference. We evaluate the proposed estimator's finite-sample performance via extensive simulations. The proposed estimator and inferential tools are implemented in the freely available R package ``ivsacim'' \citep{ying2020ivsacim} on R CRAN. We apply the proposed approach to evaluate the exposure effect of Nucleoside Reverse Transcriptase Inhibitors (NRTIs) on a key safety outcome (time to first severe or worse sign or symptom) in participants receiving a new antiretroviral regimen that omitted or added NRTIs in the open-label Optimized Treatment That Includes or Omits NRTIs (OPTIONS) trial \citep{tashima2015hiv}.

We give a preview of the data we aim to analyze since throughout the paper we constantly refer to it. 
The Optimized Treatment That Includes or Omits NRTIs trial was a multicenter, open-label, prospective, randomized, controlled study evaluating the benefits and risks of omitting versus adding NRTIs to a new optimized antiretroviral regimen for HIV-infected patients. Participants were randomly assigned either to omit or to add NRTIs after choosing an optimized regimen and an NRTI regimen. Treatment switching was present in this trial due to the potential discontinuation of NRTI assignment, which occurred when a participant in the omit-NRTIs group started any NRTIs or when a participant in the add-NRTIs group failed to initiate or permanently discontinued all NRTIs. Therefore, both directions of treatment switching could occur. \citet{ying2022structural} investigated the treatment effect on a safety outcome, the time to first severe or worse sign or symptom, under their SCSM leveraging initial randomized as an IV, assuming the exclusion restriction. However, a caveat here is that the OPTIONS trial was open-label, that is, site investigators and participants knew the treatment assignments. This might lead to the violation of the exclusion restriction, thus invalidating their analysis, which has motivated our study to alleviate such concern by proposing a robust IV method and reanalyzing the OPTIONS trial without assuming the exclusion restriction.

The remainder of the article is organized as follows. We introduce notation, structural cumulative survival models, and assumptions in Section \ref{sec:pre}. 
In Section \ref{sec:est}, we construct a recursive estimator under the SCSM and assumptions listed in Section \ref{sec:pre}. We conduct extensive simulations to evaluate finite-sample performance of our estimator in Section \ref{sec:simu} and we apply our proposed estimator to reanalyze the OPTIONS trial in Section \ref{sec:real}. Proofs and additional theoretical results are provided in the appendix.

\section{Preliminaries}\label{sec:pre}
Define
\begin{itemize}
\item $T$, a time to event outcome of interest, $C$, potential censoring time, $X = \min(T, C)$, a subject's censored event time, and $\Delta = \mathbbm{1}(T\leq C)$, a subject's observed event indicator;  
\item We introduce the counting process notation. We write $N(t) = \mathbbm{1}(X \le t, \Delta = 1)$ and $Y(t) = \mathbbm{1}(X \ge t)$ as the observed counting process and the associated at-risk processes;
\item We assume that recorded data on treatment do not change except at discrete times $\{1, \cdots, M\}$. Thus the time-varying treatment $D(m) = 1$ if subject $i$ is treated or exposed at time $m$, 0 otherwise. We write $\bar D(m) = \{D(l): 1 \leq l \le m\}$. For any $m > X$, $D(m)$ is not observed. We define $D(m) = 0$ for $m > X$, so that the whole treatment process is well defined for each subject even after the outcome event has occurred. Here we assume treatment is binary only to facilitate exposition and suit our data application. Our framework can accommodate any treatment structure like multi-level and continuous;
\item $Z$ denotes the instrumental variable corresponding to the randomized treatment assignment in a randomized trial;
\item We define $T\{\bar d(m), 0, z\}$, the potential time to event had possibly contrary to fact, the subject been assigned to treatment $z$, followed the treatment regime $\bar d(m)$ up to time $m$ and the control treatment thereafter. We make the consistency assumption that $T = T\{\bar d(m), 0, z\}$ with probability one for individuals with observed $\bar D(m) = \bar d(m)$, $D(l) = 0$, for $l > m$, and $Z = z$. We further assume that intervening on exposure can only affect survival after the time of that exposure, in other words, the event $T\{\bar D(m - 1), 0, Z\} \geq m$ occurs if and only if the event $T\{\bar D(l), 0, Z\} \geq m$ also occurs for all $l \geq m$. It follows that $\{T \geq t\}$ and $\{T\{\bar D(m), 0, Z\} \ge t\}$ are the same events for $t \in [m, m + 1)$.
\end{itemize}
We later propose an instrumental variable estimator under the handle treatment switching under a structural cumulative survival model (SCSM) without assuming exclusion restriction. We first introduce SCSMs. An SCSM imposes models on the negative log-ratio $\gamma_m\{t;\bar D(m), Z\}$ of two counterfactual survival probabilities at time $t$ under treatment strategies that differ only at time $m$ for each time $m < t$, or equivalently
\begin{align*}
    \exp\left[-\gamma_m\{t;\bar D(m), Z\}\right] = \frac{\P\left[T\{\bar D(m), 0, Z\} > t|\bar D(m), Z, T \ge m\right]}{\P\left[T\{\bar D(m - 1), 0, z = 0\} > t|\bar D(m), Z, T \ge m\right] } ,
\end{align*}
for any $t \ge m$. 
An SCSM therefore may be interpreted as encoding for individuals still at risk for the outcome at time $m$ with covariates, IV and treatment history $Z, \bar{D}(m)$, the ratio of survival probabilities of remaining event-free at time $t \ge m$ upon receiving one final blip of treatment at time $m - 1$ versus at time $m$. 

We prefer SCSMs over the rank-preserving structural failure time models (RPSFTMs) \citep{robins1991correcting} and structural cumulative failure time models (SCFTMs) \citep{shi2021instrumental}. This is because estimation under RPSFTMs requires artificial censoring to address administrative censoring, which can further increase bias, aggravate efficiency loss, and render estimation computationally challenging \citep{robins1991correcting, joffe2001administrative, joffe2012g, vansteelandt2014structural, latimer2019causal}. On the other hand, SCSMs do not require artificial censoring because it operates on cumulative survival models rather than directly dealing with actual survival times. Another important advantage of SCSMs is that the estimators based on SCSMs typically have a closed form via a recursive solution and therefore are guaranteed to exist. In contrast, the proposed estimating equation under both RPSFTMs and SCFTMs may not admit a solution, even when it does, such a solution may not be uniquely defined.

We choose to assume
\begin{equation}
    \gamma_m(t;\bar D(m), Z) = \gamma_m(t; D(m), Z) = \int_m^{t \wedge (m + 1)} D(m) dB_D(s) + ZdB_Z(s). \label{eq:condstructuralinvalid}
\end{equation}
This posits:
\begin{enumerate}
    \item The negative log-ratio $\gamma_m\{t;\bar D(m), Z\}$ is encoded by two parameters $B_D(t)$ and $B_Z(t)$ capturing the actual treatment effect by $\bar D(t)$ and the direct effect of the knowledge of the treatment assignment $Z$, respectively. Allowing $B_Z(t)$ to exist is one main contribution of this paper.
    \item The effect is allowed to vary across time nonparametrically, which is flexible to capture the possible time-varying effects of $\bar D$ and $Z$;
    \item The effect at time $m$ is short-lived in the sense that $\gamma_m(t; \bar D(m), Z)$ = $\gamma_m(m + 1; D(m), Z)$ for all $t \ge m + 1$, in other words, only the treatment $D(m)$ at time $m$ can modify the survival between $[m, m + 1)$. This assumption is needed for identification given only one instrument, which was also assumed in \citet{ying2022structural, shi2021instrumental, robins1991correcting}.
\end{enumerate}
Besides being interpreted as a ratio of conditional survival probabilities, the model specification by \eqref{eq:condstructuralinvalid} delivers another convenient interpretation of the contrast of the marginal survival functions, with an additional condition of no-current treatment value interaction \cite{robins1994adjusting}, 
\begin{align}
    &\frac{\P\left[T\{\bar D(m - 1),d(m), 0, Z\} > t|\bar D(m - 1),  D(m) = d(m), Z, T \ge m\right]}{\P\left[T\{\bar D(m - 1), 0, z = 0\} > t|\bar D(m - 1),  D(m) = d(m), Z, T \ge m\right] } \\
    &=\frac{\P\left\{T\{\bar D(m - 1), d(m), 0, Z\} > t|\bar D(m - 1), D(m) = 0, Z = 0, T \ge m\right]}{\P\left[T\{\bar D(m - 1), 0, z = 0\} > t|\bar D(m - 1), D(m) = 0, Z = 0, T \ge m\right] }. \label{eq:nocurrenttrtinteraction}
\end{align}   
It essentially states that the instantaneous causal effect of one final blip of treatment at time $m$ among individuals who were treated at time $m$ is equal to that among individuals who were not treated at time $m$ conditional on past history. In fact, under \eqref{eq:condstructuralinvalid} and the no-current treatment value interaction condition \eqref{eq:nocurrenttrtinteraction}, one can show that 
\begin{align}
    \frac{\P\left\{T(\bar d = \bar 1, z) > t\right\}}{\P\left\{T(\bar d = \bar 0, z) > t\right\}} 
    = \exp\left\{-\int_{0}^{t} 1 dB_D(s)\right\}= \exp\left\{-B_D(t)\right\}.
\end{align}
Therefore our estimand $B_D(t)$ can be interpreted as the difference in the log-marginal cumulative intensity function comparing always-treated versus never-treated regimes up to time $t$ with all subjects randomized to either treatment or control, which encodes the causal effect of interest, that is, the causal effect of treatment received. Also, we have
\begin{align}
    \frac{\P\left\{T(\bar d, z = 1) > t\right\}}{\P\left\{T(\bar d, z = 0) > t\right\}} 
    = \exp\left\{-\int_{0}^{t} 1 dB_Z(s)\right\}= \exp\left\{-B_Z(t)\right\},
\end{align}
where $B_Z(t)$ can be interpreted as the controlled direct effect \citep{vanderweele2011controlled, pearl2022direct} of the knowledge of the treatment assignment by clinicians and patients on the outcome.

As a special case of this model specification \eqref{eq:condstructuralinvalid} and often is of interest, one may further impose that the treatment effects are constant as a function of the timing of the final treatment blip, that is, setting
\begin{equation}\label{eq:const}
    \gamma_m\{t;\bar D(m), Z\} = \beta_D D(m)(t - m) + \beta_ZZ(t - m),
\end{equation}
which later we refer to as the ``constant hazards difference model''. This model encodes the SCSM analog of the ``common treatment effect'' assumption of the rank preserving structural failure time models \citep{robins1991correcting}, which states that the treatment effect is the same for all individuals (with respect to time spent on treatment) regardless of when treatment is received. 

Note that both SCSMs \eqref{eq:condstructuralinvalid} and \eqref{eq:const} described above are guaranteed to be correctly specified under the null hypothesis of no treatment effect, an appealing robustness property of the proposed framework.

Our proposed strategy leverages the randomization process as an instrumental variable satisfying two key standard IV assumptions, but not the exclusion restriction:
\begin{assump}[IV relevance]\label{assump:ivrelevance}
The instrument is associated with the exposure at $m$ for individuals still  at risk for the event time for all $m$; specifically, 
\begin{equation}
Z \not\perp D(m)~|~T \ge m, \bar D(m - 1).
\end{equation}
\end{assump}
IV relevance requires that for subjects who remain at risk for the outcome event at time $m$, the instrument remains predictive of current treatment status even after conditioning on treatment and covariate history. This is typically a reasonable assumption in a randomized trial, given that individuals randomized to the active arm are more likely than in the control arm to be treated over time, even upon conditioning on their history.
\begin{assump}[IV independence]\label{assump:ivindependence}
The instrument variable is independent of the potential outcome under no treatment, 
\begin{equation}
Z \perp T(\bar d = \bar 0, z = 0).
\end{equation}
\end{assump}
IV independence ensures that the initial randomization itself is unconfounded, which is apparently satisfied in a randomized trial, no matter double-blinded or open-label.

For completeness, here we introduce the exclusion restriction typically assumed in the IV literature that we do not impose. The exclusion restriction assumes that $T(\bar d, z) = T(\bar d)$, which rules out the possibility that randomization itself can impact the outcome via a pathway not involving treatment actually taken. However, as we mentioned earlier, this assumption is inclined to fail in an open-label randomized trial, as the knowledge of assigned treatment may change clinicians' treatment guidelines or patients' health-seeking behavior, which in turn influences the outcome. \citet{ying2022structural} not only imposed model \eqref{eq:condstructuralinvalid}, Assumptions \ref{assump:ivrelevance}, \ref{assump:ivindependence} but also assumed the exclusion restriction and set $B_Z(t) \equiv 0$ in \eqref{eq:condstructuralinvalid}, leading to
\begin{align}
    \frac{\P\left[T\{\bar D(m), 0\} > t|\bar D(m), Z, T \ge m\right]}{\P\left[T\{\bar D(m - 1), 0\} > t|\bar D(m), Z, T \ge m\right] } = \exp\left[-\int_m^{t \wedge (m + 1)} D(m) dB_D(s)\right],\label{eq:ytt}
\end{align}
which asserts that the instrument $Z$ does not modify the effect of the treatment $\bar D(m)$ on the outcome $T$ on the multiplicative scale. We, however, do not assume them. 

The relaxation of the exclusion restriction  leads to another unknown parameter $B_Z(t)$ compared to that in \citet{ying2022structural} which requires an additional estimating equation to identify. Inspired by \citet{tchetgen2021genius}, we impose the following assumption leveraging the treatment process as:
\begin{assump}[IV No Interaction with Unmeasured Confounders]\label{assump:ivnointeraction}
We assume the treatment process satisfies
\begin{equation}\label{eq:propscore}
    \E\{D(t)|Z, T(\bar d = \bar 0, z = 0)\} = \alpha_1(t; Z) + \alpha_2\{t; T(\bar d = \bar 0, z = 0)\}.
\end{equation}
\end{assump}
Such a linear model for treatment process with no interaction between $Z$ and $T(\bar d = \bar 0, z = 0)$ is typically assumed in IV literature \citep{li2015instrumental, tchetgen2015instrumental, ying2019two}, where they introduced an unknown confounder $U$ for convenience, playing a similar role here as $T(\bar d = \bar 0, z = 0)$.

We make a standard conditional independent censoring assumption.
\begin{assump}[Conditional independent censoring]\label{assump:condindcensoring}
\begin{equation}
C \perp (T, \bar D(t), Z).
\end{equation}
\end{assump}
This censoring assumption simplifies estimation so that we can concentrate on dealing with treatment switching. Note that although not further pursued here, the above assumption can be relaxed substantially by only requiring that $C \perp T ~|~\bar D(m), Z, X \geq m$. This, on the other hand, requires further adjustment for dependent censoring possibly by standard inverse probability censoring weighting.

\section{Estimation}\label{sec:est}
Suppose we observe $n$ independent and identically distributed samples. Define $\E$ and $\E_n$ as the population mean and the empirical mean. Also we write $Z^c$ as $Z - \E_n(Z)$ and $D(t)^c$ as $D(t) - \E_n\{D(t)|Z\}$. $\E_n\{D(t)|Z\}$ is nonparametric because $Z$ is binary. We construct an explicit and recursive estimator 
\begin{equation}\label{eq:empiricalest}
\begin{pmatrix}
\hat B_D(t)\\
\hat B_Z(t)
\end{pmatrix}
= \int_0^t\{\hat \bbM(s)\}^\dagger \E_n \left[
\begin{pmatrix}
Z^c\\
Z^cD(s)^c
\end{pmatrix}e^{\int_0^{s-}D(u)d\hat B_D(u) + Zd\hat B_Z(u)}dN(s)\right],
\end{equation}
where
\begin{align}
    \hat \bbM(s)
    = \E_n \left[
    \begin{pmatrix}
    Z^cD(s) & Z^cD(s)^cD(s)\\
    Z^cZ &Z^cD(s)^cZ
    \end{pmatrix}
    Y(s)e^{\int_0^{s-}D(u)d\hat B_D(u) + Zd\hat B_Z(u)}\right],
\end{align}
and $(\cdot)^\dagger$ denotes the Moore-Penrose generalized inverse. Because of the recursive structure of $\hat B(t)$ in \eqref{eq:empiricalest}, and the key fact that the estimator only changes values at observed event times, we may evaluate it forward in time, with initial value $\hat B_D(0) = \hat B_Z(0) = 0$. 

We also propose an estimator under the constant hazards difference model \eqref{eq:const}, where one may use
\begin{equation}\label{eq:constantest}
    \hat \beta_D = \int_0^\tau w(t)d\hat B_D(t), ~~~~\hat \beta_Z = \int_0^\tau w(t)d\hat B_Z(t),
\end{equation}
with $w(t) = \tilde w(t)/\int_0^\tau \tilde w(s)ds$, $\tilde w(t) = \E_n\{Y(t)\}$, and $\tau$ denoting time of end of study. Note that although all theorems below are built for $\hat B_D(t)$, they can be immediately translated for $\hat \beta_D$ by Slutsky's theorem and the functional Delta method, which we omit the details.

\begin{thm}\label{thm:cons}
Under model \eqref{eq:condstructuralinvalid}, Assumptions \ref{assump:ivrelevance} - \ref{assump:condindcensoring}, and regularity Assumptions \ref{assump:uniqsolu} - \ref{assump:ivbound} in the appendix, the estimator $(\hat B_D(t), \hat B_Z(t))$ is uniformly consistent for $B_D(t)$ on $[0, \tau]$, namely,
\begin{equation}
\sup_{t \in [0, \tau]} \left|(\hat B_D(t), \hat B_Z(t)) - (B_D(t), B_Z(t))\right| \to 0~~ \text{a.s.}.
\end{equation}
The normalized process 
\begin{equation}
    \sqrt{n}\{(\hat B_D(t), \hat B_Z(t)) - (B_D(t), B_Z(t))\}
\end{equation}
converges weakly to a two-dimensional zero-mean Gaussian process.
\end{thm}
We postpone the analytic variance estimate of the asymptotic Gaussian process to the appendix due to its complicated structure. Nonetheless, we have implemented our estimator together with its inferential tools into the R package named ``ivsacim'' \citep{ying2020ivsacim} freely available on R CRAN. Inferential tools include estimate of standard deviations, Z-values, P-values, a goodness-of-fit test for the constant hazards difference model \eqref{eq:const} ($H_0:~B_D(t) = \beta_Dt$ for all $t$) and as well as for the causal null hypothesis ($H_0:~B_D(t) \equiv 0$ for all t).

\section{Simulation}\label{sec:simu}
\noindent
\textbf{Aims: }To investigate the finite-sample performance of our proposed estimators $\hat B_D(t)$ in \eqref{eq:empiricalest} and $\hat \beta_D$ in \eqref{eq:constantest} which are of primary interest. 

\noindent
\textbf{Data-generating mechanisms: }
In order to investigate the finite sample performance  of our proposed methods, we conduct a simulation study in which we generate  $B = 1000$ data sets of i.i.d data with sample size $N =1600, 3200$. We generate a bivariate baseline variable $U$ which confounds the relationship between time-varying treatment and the time-to-event outcome:
\begin{equation}
U = (U_{1}, U_{2})^\top \sim \cN\left\{\begin{pmatrix}
3/2\\
3/2
\end{pmatrix},
\begin{pmatrix}
1/4 & -1/6\\
-1/6 &1/4
\end{pmatrix}\right\}.
\end{equation}
We simulate a scenario in which initial treatment assignment $Z$ is generated as an independent Bernoulli random variable with event probability $\P(Z = 1) = 0.5$. We also generate a potential treatment switching time $W$ for each individual according to 
\begin{equation}\label{eq:timetoswitch}
    \P(W > t|Z, U) = \exp(-0.5t) + Z\{1 - \exp(-0.05t)\} + (2Z - 1)\{1 - \exp(-0.05t -0.1 \cdot U_{1} t)\},
\end{equation}
and discretize it into a grid with step size $=0.1$. A subject experiences treatment $Z$ before $W$ and is switched to $1 - Z$ right after $W$. Thus, $Z$ and $W$ determine a patient's entire treatment process $\bar D$. By \eqref{eq:timetoswitch} we also allow both directions of treatment switching, which fits our application setting. The potential time to event $\tilde T(\bar d)$ is generated according to
\begin{align}
    \P\{\tilde T(\bar d, z) > t|U\} = \exp\left\{-0.1 \cdot t - 0.2 \cdot \int_0^td(s)ds - 0.1 \cdot \int_0^tzds  - 0.15 \cdot U_{2} \cdot t\right\},
\end{align}
and the observed time to event $T$ is thus generated via consistency with $\bar D$. Independent censoring was then generated with an overall rate of $18\%$. Treatment switching occurred at an approximate rate of 14\%. In the appendix, we confirm that under the proposed data generating mechanism, \eqref{eq:condstructuralinvalid} and Assumption \ref{assump:ivindependence}-\ref{assump:condindcensoring} hold. Indeed, one can show that \eqref{eq:condstructuralinvalid} holds as
\begin{align}
    \frac{\P\{\tilde T(\bar D(m), 0, Z) > t|\bar D(m), Z, \tilde T > m\}}{\P\{\tilde T(\bar D(m - 1), 0, z = 0) > t|\bar D(m), Z, \tilde T > m\}} = \exp\left\{- 0.2 \cdot \int_{m}^{t \wedge (m + 1)}D(s)ds - 0.1 \cdot \int_{m}^{t \wedge (m + 1)}Zds\right\}.
\end{align}

\noindent
\textbf{Estimands: }Our estimands are the SCSM treatment effect $B_D(t) = 0.2 t$ and the constant hazards difference effect $\beta_D = 0.2$.

\noindent
\textbf{Methods:} Each simulated dataset is analyzed using $\hat B_{D}(t)$ in \eqref{eq:empiricalest} and $\hat \beta_{D}$ in \eqref{eq:constantest}. As a comparison, we also fit the time-varying treatment effect estimator $\hat B_{D, \text{YTT}}(t)$ and the constant effect estimator $\hat \beta_{D, \text{YTT}}$ in \citet{ying2022structural} assuming \eqref{eq:ytt} and the exclusion restriction were to hold. $\hat B_{D, \text{YTT}}(t)$ and $\hat \beta_{D, \text{YTT}}$ are expected to be biased.

\noindent
\textbf{Performance measures: }We report bias, empirical standard errors (SEE), average estimated standard errors (SD), and coverage probabilities of 95\% confidence intervals of both $\hat B_{D}$ and $\hat B_{D, \text{YTT}}$ at time $t = 1, 2, 3$, also $\hat \beta_{D}$, $\hat \beta_{D, \text{YTT}}$.

Simulation results concerning $\hat B_{D}(t)$ \eqref{eq:empiricalest} and $\hat \beta_{D}$ \eqref{eq:constantest} are given in Table \ref{tab:invalidsimuresults}. Simulation results concerning $\hat B_{D, \text{YTT}}(t)$ \eqref{eq:empiricalest} and $\hat \beta_{D, \text{YTT}}$ \eqref{eq:constantest} are given in Table \ref{tab:validsimuresults}. The results confirm that the proposed estimator $\hat B_{D}(t)$ has small biases both at sample sizes 1600 and 3200 at $t = 1, 2, 3$. The estimated standard errors match Monte Carlo standard errors and overall 95\% confidence intervals attain the nominal levels among sample sizes 1600 and 3200 at $t = 1, 2, 3$. It is seen from the simulation that the estimator $\hat \beta_{D}$ is consistent and that the variability is well estimated, leading to satisfactory coverage probabilities. \citet{ying2022structural} has shown its method failed under a violation of exclusion restriction in its appendix. Here, our simulation once again confirmed this. In fact, $\hat B_{D, \text{YTT}}(t)$ is severely biased regardless of the sample size, especially when $t = 2, 3$. This is because more subjects switch their treatment and the direct effect of the initial treatment assignment starts to influence the estimator. The coverage rates of $\hat B_{D, \text{YTT}}(t)$ fail to attain the nominal level due to biases. The same situation applies to $\hat \beta_{D, \text{YTT}}$ as well.

\begin{table}[H]
\centering
\caption{Simulation results for the SCSM treatment effect model $\hat B_{D}(t)$ and the constant hazards difference model $\hat \beta_{D}$. Bias of $\hat B_{D}(t)$, empirical standard error, see($\hat B_{D}(t)$), average estimated standard error, sd($\hat B_{D}(t)$), and coverage probability of 95\% confidence intervals CP($\hat B_{D}(t)$), 95\% CP($\hat B_{D}(t)$), at time $t = 1, 2, 3$, bias of $\hat \beta_{D}$, empirical standard error, see($\hat \beta_{D}$), average estimated standard error, sd($\hat \beta_{D}$), and coverage probability of 95\% pointwise confidence intervals CP($\hat \beta_{D}$), 95\% CP($\hat \beta_{D}$), for sample size $N=1600, 3200$ and $R = 1000$ Monte Carlo samples.}
\label{tab:invalidsimuresults}
\resizebox{\columnwidth}{!}{%

\begin{tabular}{|c|l|l|l|l|l|l|}
\hline
  Sample Sizes           & \multicolumn{1}{l|}{}     & t = 1   & t = 2   & t = 3    & & \\ \hline
\multirow{4}{*}{N = 1600} 
& Bias($\hat B_{D}(t)$)    & -0.0186  & -0.0233  & -0.0167  & Bias($\hat \beta_{D}$) &-0.0058\\
& SEE($\hat B_{D}(t)$)     & 0.0918  & 0.1141  & 0.1734 & see($\hat \beta_{D}$) &0.0591\\ 
&SD($\hat B_{D}(t)$)       & 0.1125  & 0.1526  & 0.2223 & sd($\hat \beta_{D}$) &0.0738\\ 
&95\% CP($\hat B_{D}(t)$)  & 94.0  & 96.8  & 96.5 & 95\% CP($\hat \beta_{D}$) &95.3\\ \hline
\multirow{4}{*}{N = 3200} 
&Bias($\hat B_{D}(t)$)      & -0.0133 & -0.0140 & -0.0145 & Bias($\hat \beta_{D}$) &-0.0064\\
&SEE($\hat B_{D}(t)$)       &0.0623  & 0.0852  & 0.1331 & see($\hat \beta_{D}$) &0.0441\\ 
&SD($\hat B_{D}(t)$)        &0.0781  & 0.1051  & 0.1556 & sd($\hat \beta_{D}$) &0.0515\\ 
&95\% CP($\hat B_{D}(t)$)   &95.3  & 95.8  & 95.5 &95\% CP($\hat \beta_{D}$) &95.9 \\ \hline
\end{tabular}
}
\end{table}

\begin{table}[H]
\centering
\caption{Simulation results for the SCSM treatment effect model $\hat B_{D, \text{YTT}}(t)$ and the constant hazards difference model $\hat \beta_{D, \text{YTT}}$. Bias of $\hat B_{D, \text{YTT}}(t)$, empirical standard error, see($\hat B_{D, \text{YTT}}(t)$), average estimated standard error, sd($\hat B_{D, \text{YTT}}(t)$), and coverage probability of 95\% confidence intervals CP($\hat B_{D, \text{YTT}}(t)$), 95\% CP($\hat B_{D, \text{YTT}}(t)$), at time $t = 1, 2, 3$, bias of $\hat \beta_{D, \text{YTT}}$, empirical standard error, see($\hat \beta_{D, \text{YTT}}$), average estimated standard error, sd($\hat \beta_{D, \text{YTT}}$), and coverage probability of 95\% pointwise confidence intervals CP($\hat \beta_{D, \text{YTT}}$), 95\% CP($\hat \beta_{D, \text{YTT}}$), for sample size $N=1600, 3200$ and $R = 1000$ Monte Carlo samples.}
\label{tab:validsimuresults}
\resizebox{\columnwidth}{!}{%

\begin{tabular}{|c|l|l|l|l|l|l|}
\hline
  Sample Sizes           & \multicolumn{1}{l|}{}     & t = 1   & t = 2   & t = 3    & & \\ \hline
\multirow{4}{*}{N = 1600} 
& Bias($\hat B_{D, \text{YTT}}(t)$)    & 0.0543 & 0.1227 & 0.2185     & Bias($\hat \beta_{D, \text{YTT}}$) &0.0704\\
& SEE($\hat B_{D, \text{YTT}}(t)$)     &0.0385  & 0.0755  & 0.1439    & see($\hat \beta_{D, \text{YTT}}$) &0.0387\\ 
&SD($\hat B_{D, \text{YTT}}(t)$)       & 0.0435  & 0.0861 & 0.1687   & sd($\hat \beta_{D, \text{YTT}}$) &0.0451\\ 
&95\% CP($\hat B_{D, \text{YTT}}(t)$)  & 75.9  & 71.9  & 80.0         & 95\% CP($\hat \beta_{D, \text{YTT}}$) &69.2\\ \hline
\multirow{4}{*}{N = 3200} 
&Bias($\hat B_{D, \text{YTT}}(t)$)      & 0.0546 & 0.1265 & 0.2241     & Bias($\hat \beta_{D, \text{YTT}}$) &0.0723\\
&SEE($\hat B_{D, \text{YTT}}(t)$)       &0.0261  & 0.0519  & 0.0970   & see($\hat \beta_{D, \text{YTT}}$) &0.0271\\ 
&SD($\hat B_{D, \text{YTT}}(t)$)        &0.0308  & 0.0608  & 0.1187    & sd($\hat \beta_{D, \text{YTT}}$) &0.0317\\ 
&95\% CP($\hat B_{D, \text{YTT}}(t)$)   &58.7  & 43.5 & 54.3         &95\% CP($\hat \beta_{D, \text{YTT}}$) &36.1 \\ \hline
\end{tabular}
}
\end{table}

\section{Real Data Application}\label{sec:real}
We aim to reanalyze the treatment effect on a safety outcome, the time to first severe or worse sign or symptom, in the Optimized Treatment That Includes or Omits NRTIs trial as \citet{ying2022structural} did. An introduction of the analysis was given in Section \ref{sec:intro}. A summary of events and treatment switching in the study can be found in Table \ref{tab:outcomesummary}. \citet{ying2022structural} investigated the treatment effect under their SCSM leveraging initial randomization as an IV, assuming the exclusion restriction. They found a significant time-varying effect of NRTIs on the safety outcome, which revealed possible safety concerns of NRTIs. However, a caveat here is that the OPTIONS trial was open-label, that is, site investigators and participants knew the treatment assignments. This might lead to the violation of the exclusion restriction, thus invalidating their analysis. Therefore, we would like to reanalyze the OPTIONS trial without assuming the exclusion restriction.

\begin{table}[H]
\centering
\caption{A summary of the safety outcome (first severe or worse sign or symptom) and NRTI assignment change (treatment switching) in the OPTIONS trial, compared between treatment groups. \label{tab:outcomesummary}}{
\begin{tabular}{|l|ll|l|}
\hline
\multirow{2}{*}{}                                                                                  & \multicolumn{2}{l|}{Patients, n (\%)}                                                                                                                 & \multirow{2}{*}{\begin{tabular}[c]{@{}l@{}}Difference (95\% CI),\\ percentage points\end{tabular}} \\ \cline{2-3}
                                                                                                   & \multicolumn{1}{l|}{\begin{tabular}[c]{@{}l@{}}Add NRTIs\\  (n = 180)\end{tabular}} & \begin{tabular}[c]{@{}l@{}}Omit NRTIs\\  (n = 177)\end{tabular} &                                                                                                    \\ \hline
First severe or worse sign or symptom                                                              & \multicolumn{1}{l|}{51 (28.3)}                                                      & 35 (19.8)                                                       & 8.6 (-0.8 to 17.9)                                                                                 \\ \hline
Change in NRTI assignment                                                                          & \multicolumn{1}{l|}{10 (5.3)}                                                       & 17 (9.5)                                                        & -4.0 (-10.1 to 2.0)                                                                                \\ \hline
\begin{tabular}[c]{@{}l@{}}Change in NRTI assignment \\ before safety outcome failure\end{tabular} & \multicolumn{1}{l|}{8 (3.4)}                                                        & 9 (2.8)                                                         & -0.6 (-4.4 to 5.4)                                                                                 \\ \hline
\end{tabular}
}
\end{table}

Using the proposed approach to formally account for treatment switching by leveraging randomized treatment assignment as an IV for treatment actually received which is likely confounded by unmeasured factors, we performed a test of the sharp null hypothesis of no individual causal effect, i.e. $B_D(t) = 0$, against which we found nonsignificant statistical evidence, P-value 0.086. Our approach also delivered a nonparametric estimator $\hat B_D(t)$ along with 95\% pointwise confidence bands displayed at the top left in Figure \ref{fig:time2firstsafety}. From the figure, we observe a hazard rate for experiencing the safety outcome severe/worse sign/symptom over time in the add-NRTI group compared to the omit-group. Under a constant hazards difference model \eqref{eq:const}, our approach estimated a hazards difference of -0.0537 (-0.120, 0.012), P-value 0.109, though there is significant evidence of a time-varying effect, indicated by our goodness-of-fit test rejecting the constant effect model (P-value 0.039). The test of the null hypothesis of no direct causal effect of the randomization, i.e. $B_Z(t) = 0$, reports a nonsignificant P-value 0.08. A nonparametric estimator $\hat B_Z(t)$ along with 95\% pointwise confidence bands are displayed at the top right in Figure \ref{fig:time2firstsafety}.


\begin{figure}[H]
\centering
\includegraphics[scale = 0.4]{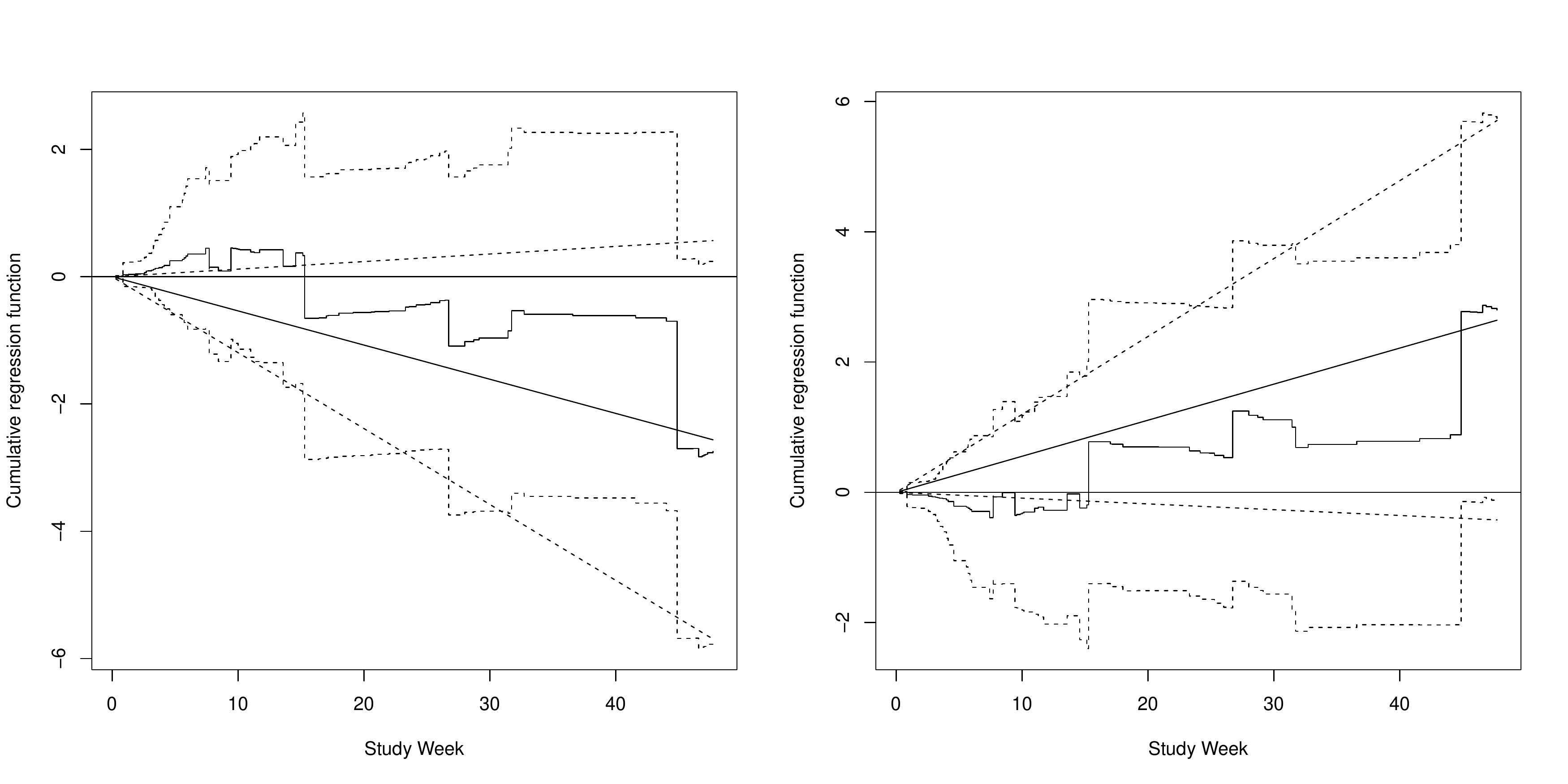}
\includegraphics[scale = 0.32]{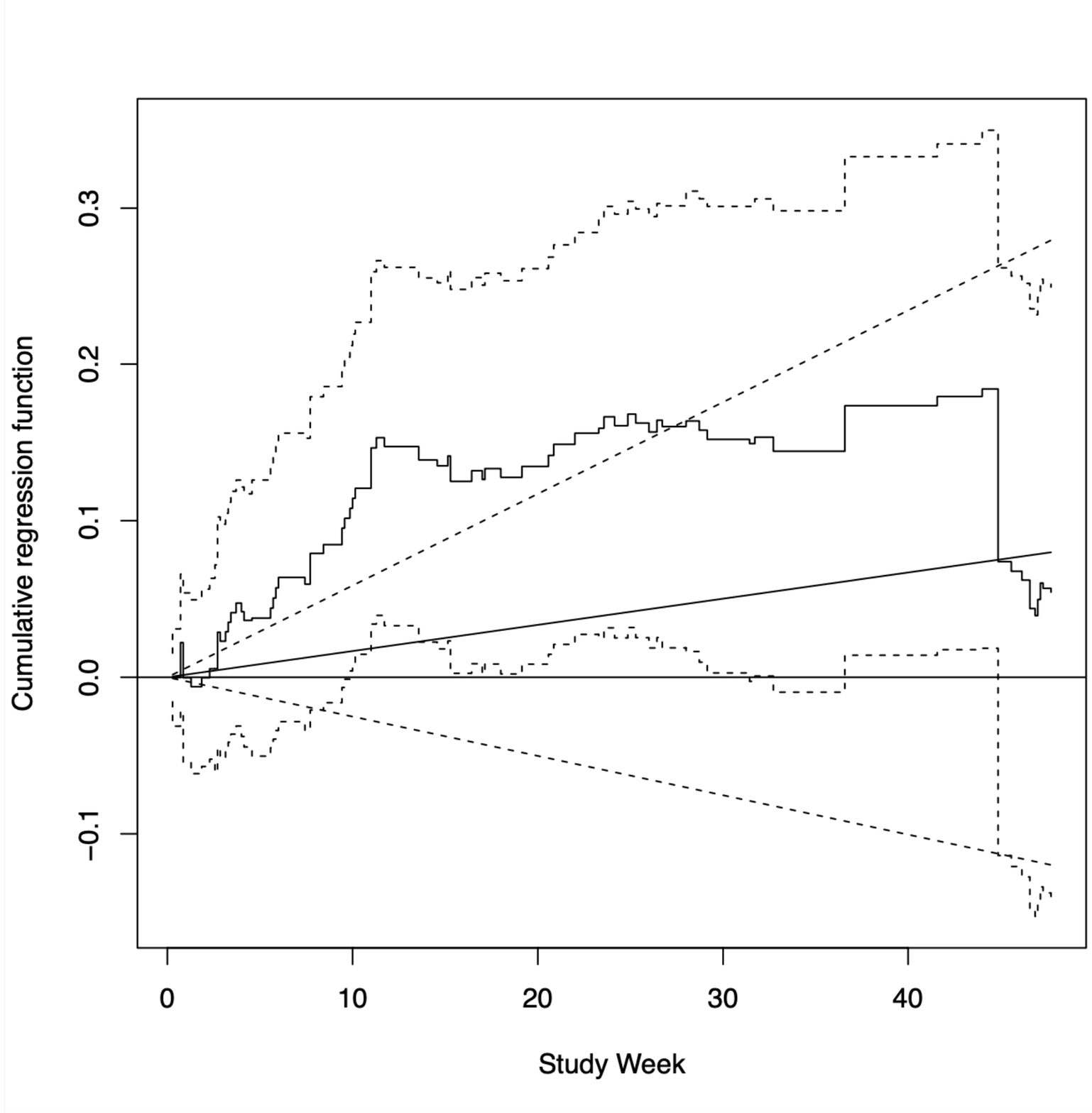}
\caption{Top left figure: The step line corresponds to the estimated SCSM treatment effect of NRTIs, $\hat B_D(t)$, on the safety outcome (time to first severe or worse sign or symptom), along with 95\% pointwise confidence bands. The y-axis corresponds to the estimated SCSM treatment effect at time t, namely, $\hat B_D(t)$. The curved line corresponds to the constant hazards difference estimator $\hat \beta_D$. Top right figure: The step line corresponds to the estimated SCSM direct effect of initial treatment assignment, $\hat B_Z(t)$, on the safety outcome, along with 95\% pointwise confidence bands. The y-axis corresponds to the estimated SCSM treatment effect at time t, namely, $\hat B_Z(t)$. The curved line corresponds to the constant hazards difference estimator $\hat \beta_Z$. Bottom figure: The step line corresponds to the estimated SCSM total effect of treatment assignment, $\hat B_D(t) + \hat B_Z(t)$, on the safety outcome, along with 95\% pointwise confidence bands. The y-axis corresponds to the estimated SCSM treatment effect at time t, namely, $\hat B_D(t) + \hat B_Z(t)$. The curved line corresponds to the constant hazards difference estimator $\hat \beta_D + \hat \beta_Z$. }
\label{fig:time2firstsafety}
\end{figure}

\citet{ying2022structural} uncovered possible safety concerns of NRTIs, that is, adding NRTIs may increase the risk of severe or worse sign or symptom for HIV-infected patients. Our analysis reveals that such safety concern is possibly caused by the effect of unblinded treatment assignment. Therefore one possible explanation is that knowing treatment assignments had possibly led to a change of treatment courses by the clinicians or health-seeking behavior of treatment-aware patients, which resulted in worse sign or symptom. We also append $\hat B_D(t) + \hat B_Z(t)$ along with 95\% pointwise confidence bands displayed at the bottom in Figure \ref{fig:time2firstsafety}, which indeed shows a similar trend as the real data application in \citet{ying2022structural}. One possible explanation is that since \citet{ying2022structural} cannot distinguish the effect of $D(t)$ and $Z$ but can only capture the the total effect of treatment actually received and assigned $B_D(t) + B_Z(t)$ when the no-current treatment value interaction condition \eqref{eq:nocurrenttrtinteraction} holds.



\section*{Acknowledgement}
The author would like to thank Eric J. Tchetgen Tchetgen for his valuable suggestions for improving the manuscript. The author would also like to thank Dr. Diana Ventura from the Center for Biostatistics in AIDS Research at T.H. Chan School of Public Health at Harvard University for providing us with the raw data for this manuscript. \vspace*{-8pt}

\bibliographystyle{agsm}
\bibliography{ref}

@article{angrist1996identification,
    Author = {Angrist, Joshua D and Imbens, Guido W and Rubin, Donald B},
    Journal = {Journal of the American Statistical Association},
    Number = {434},
    Pages = {444--455},
    Publisher = {Taylor \& Francis},
    Title = {Identification of causal effects using instrumental variables},
    Volume = {91},
    Year = {1996}
}

@article{angrist2001instrumental,
  title={Instrumental variables and the search for identification: from supply and demand to natural experiments},
  author={Angrist, Joshua D and Krueger, Alan B},
  journal={Journal of Economic Perspectives},
  volume={15},
  number={4},
  pages={69--85},
  year={2001}
}

@article{cuzick1997adjusting,
  title={Adjusting for non-compliance and contamination in randomized clinical trials},
  author={Cuzick, Jack and Edwards, Robert and Segnan, Nereo},
  journal={Statistics in Medicine},
  volume={16},
  number={9},
  pages={1017--1029},
  year={1997},
  publisher={Wiley Online Library}
}

@article{demetri2006efficacy,
  title={Efficacy and safety of sunitinib in patients with advanced gastrointestinal stromal tumour after failure of imatinib: a randomised controlled trial},
  author={Demetri, George D and van Oosterom, Allan T and Garrett, Christopher R and Blackstein, Martin E and Shah, Manisha H and Verweij, Jaap and McArthur, Grant and Judson, Ian R and Heinrich, Michael C and Morgan, Jeffrey A and others},
  journal={The Lancet},
  volume={368},
  number={9544},
  pages={1329--1338},
  year={2006},
  publisher={Elsevier}
}

@article{jimenez2017evaluating,
  title={Evaluating the effects of treatment switching with randomization as an instrumental variable in a randomized controlled trial},
  author={Jimenez, Sara and Lai, Dejian and Gould, K Lance and Davis, Barry R},
  journal={Communications in Statistics-Simulation and Computation},
  volume={46},
  number={6},
  pages={4966--4990},
  year={2017},
  publisher={Taylor \& Francis}
}

@article{joffe2001administrative,
  title={Administrative and artificial censoring in censored regression models},
  author={Joffe, Marshall M},
  journal={Statistics in Medicine},
  volume={20},
  number={15},
  pages={2287--2304},
  year={2001},
  publisher={Wiley Online Library}
}

@article{joffe2012g,
  title={G-estimation and artificial censoring: problems, challenges, and applications},
  author={Joffe, Marshall M and Yang, Wei Peter and Feldman, Harold},
  journal={Biometrics},
  volume={68},
  number={1},
  pages={275--286},
  year={2012},
  publisher={Wiley Online Library}
}

@article{latimer2017adjusting,
  title={Adjusting for treatment switching in randomised controlled trials--a simulation study and a simplified two-stage method},
  author={Latimer, Nicholas R and Abrams, KR and Lambert, PC and Crowther, MJ and Wailoo, AJ and Morden, JP and Akehurst, RL and Campbell, MJ},
  journal={Statistical Methods in Medical Research},
  volume={26},
  number={2},
  pages={724--751},
  year={2017},
  publisher={SAGE Publications Sage UK: London, England}
}

@article{latimer2018assessing,
  title={Assessing methods for dealing with treatment switching in clinical trials: a follow-up simulation study},
  author={Latimer, Nicholas R and Abrams, Keith R and Lambert, Paul C and Morden, James P and Crowther, Michael J},
  journal={Statistical Methods in Medical Research},
  volume={27},
  number={3},
  pages={765--784},
  year={2018},
  publisher={SAGE Publications Sage UK: London, England}
}

@article{latimer2019causal,
  title={Causal inference for long-term survival in randomised trials with treatment switching: Should re-censoring be applied when estimating counterfactual survival times?},
  author={Latimer, NR and White, IR and Abrams, KR and Siebert, Uwe},
  journal={Statistical methods in medical research},
  volume={28},
  number={8},
  pages={2475--2493},
  year={2019},
  publisher={SAGE Publications Sage UK: London, England}
}

@article{li2015instrumental,
	Author = {Li, Jialiang and Fine, Jason and Brookhart, Alan},
	Journal = {Biometrics},
	Number = {1},
	Pages = {122--130},
	Publisher = {Wiley Online Library},
	Title = {Instrumental variable additive hazards models},
	Volume = {71},
	Year = {2015}
}

@article{lin2000semiparametric,
  title={Semiparametric regression for the mean and rate functions of recurrent events},
  author={Lin, Danyu Y and Wei, Lee-Jen and Yang, I and Ying, Zhiliang},
  journal={Journal of the Royal Statistical Society: Series B (Statistical Methodology)},
  volume={62},
  number={4},
  pages={711--730},
  year={2000},
  publisher={Wiley Online Library}
}

@book{martinussen2006dynamic,
  title={Dynamic regression models for survival data},
  author={Martinussen, Torben and Scheike, Thomas H},
  volume={1},
  year={2006},
  publisher={Springer}
}

@article{martinussen2017instrumental,
  title={Instrumental variables estimation of exposure effects on a time-to-event endpoint using structural cumulative survival models},
  author={Martinussen, Torben and Vansteelandt, Stijn and Tchetgen Tchetgen, Eric J. and Zucker, David M},
  journal={Biometrics},
  volume={73},
  number={4},
  pages={1140--1149},
  year={2017},
  publisher={Wiley Online Library}
}

@article{morden2011assessing,
  title={Assessing methods for dealing with treatment switching in randomised controlled trials: a simulation study},
  author={Morden, James P and Lambert, Paul C and Latimer, Nicholas and Abrams, Keith R and Wailoo, Allan J},
  journal={BMC Medical Research Methodology},
  volume={11},
  number={1},
  pages={1--20},
  year={2011},
  publisher={Springer}
}

@article{motzer2008efficacy,
  title={Efficacy of everolimus in advanced renal cell carcinoma: a double-blind, randomised, placebo-controlled phase III trial},
  author={Motzer, Robert J and Escudier, Bernard and Oudard, St{\'e}phane and Hutson, Thomas E and Porta, Camillo and Bracarda, Sergio and Gr{\"u}nwald, Viktor and Thompson, John A and Figlin, Robert A and Hollaender, Norbert and others},
  journal={The Lancet},
  volume={372},
  number={9637},
  pages={449--456},
  year={2008},
  publisher={Elsevier}
}

@incollection{pearl2022direct,
  title={Direct and indirect effects},
  author={Pearl, Judea},
  booktitle={Probabilistic and Causal Inference: The Works of Judea Pearl},
  pages={373--392},
  year={2022}
}

@article{robins1991correcting,
  title={Correcting for non-compliance in randomized trials using rank preserving structural failure time models},
  author={Robins, James M and Tsiatis, Anastasios A},
  journal={Communications in Statistics-Theory and Methods},
  volume={20},
  number={8},
  pages={2609--2631},
  year={1991},
  publisher={Taylor \& Francis}
}

@article{robins1994adjusting,
  title={Adjusting for differential rates of prophylaxis therapy for PCP in high-versus low-dose AZT treatment arms in an AIDS randomized trial},
  author={Robins, James M and Greenland, Sander},
  journal={Journal of the American Statistical Association},
  volume={89},
  number={427},
  pages={737--749},
  year={1994},
  publisher={Taylor \& Francis}
}

@article{robins1998structural,
  title={Structural nested failure time models},
  author={Robins, James M},
  journal={Encyclopedia of Biostatistics},
  volume={6},
  pages={4372--4389},
  year={1998},
  publisher={Wiley Chichester}
}

@article{robins2000correcting,
  title={Correcting for noncompliance and dependent censoring in an AIDS clinical trial with inverse probability of censoring weighted (IPCW) log-rank tests},
  author={Robins, James M and Finkelstein, Dianne M},
  journal={Biometrics},
  volume={56},
  number={3},
  pages={779--788},
  year={2000},
  publisher={Wiley Online Library}
}

@article{shi2021instrumental,
  title={Instrumental variable estimation for a time-varying treatment and a time-to-event outcome via structural nested cumulative failure time models},
  author={Shi, Joy and Swanson, Sonja A and Kraft, Peter and Rosner, Bernard and De Vivo, Immaculata and Hern{\'a}n, Miguel A},
  journal={BMC Medical Research Methodology},
  volume={21},
  number={1},
  pages={1--12},
  year={2021},
  publisher={BioMed Central}
}

@article{sullivan2020adjusting,
  title={Adjusting for Treatment Switching in Oncology Trials: A Systematic Review and Recommendations for Reporting},
  author={Sullivan, Thomas R and Latimer, Nicholas R and Gray, Jodi and Sorich, Michael J and Salter, Amy B and Karnon, Jonathan},
  journal={Value in Health},
  volume={23},
  number={3},
  pages={388--396},
  year={2020},
  publisher={Elsevier}
}

@article{tashima2015hiv,
  title={HIV salvage therapy does not require nucleoside reverse transcriptase inhibitors: a randomized, controlled trial},
  author={Tashima, Karen T and Smeaton, Laura M and Fichtenbaum, Carl J and Andrade, Adriana and Eron, Joseph J and Gandhi, Rajesh T and Johnson, Victoria A and Klingman, Karin L and Ritz, Justin and Hodder, Sally and others},
  journal={Annals of Internal Medicine},
  volume={163},
  number={12},
  pages={908--917},
  year={2015},
  publisher={American College of Physicians}
}

@article{tchetgen2015instrumental,
	Author = {Tchetgen Tchetgen, Eric J. and Walter, Stefan and Vansteelandt, Stijn and Martinussen, Torben and Glymour, Maria},
	Journal = {Epidemiology (Cambridge, Mass.)},
	Number = {3},
	Pages = {402-410},
	Publisher = {NIH Public Access},
	Title = {Instrumental variable estimation in a survival context},
	Volume = {26},
	Year = {2015}
}

@article{tchetgen2021genius,
  title={The GENIUS approach to robust Mendelian randomization inference},
  author={Tchetgen Tchetgen, Eric J. and Sun, BaoLuo and Walter, Stefan},
  journal={Statistical Science},
  volume={36},
  number={3},
  pages={443--464},
  year={2021},
  publisher={Institute of Mathematical Statistics}
}

@article{tsiatis2021estimating,
  title={Estimating vaccine efficacy over time after a randomized study is unblinded},
  author={Tsiatis, Anastasios A and Davidian, Marie},
  journal={arXiv preprint arXiv:2102.13103},
  year={2021}
}

@incollection{van1996weak,
	Author = {Van Der Vaart, Aad W and Wellner, Jon A},
	Booktitle = {Weak {C}onvergence and {E}mpirical {P}rocesses},
	Pages = {16--28},
	Publisher = {Springer},
	Title = {Weak convergence},
	Year = {1996}
}

@book{van2000asymptotic,
    Author = {Van der Vaart, Aad W},
    Publisher = {Cambridge university press},
    Title = {Asymptotic {S}tatistics},
    Volume = {3},
    Year = {2000}
}

@article{vanderweele2011controlled,
  title={Controlled direct and mediated effects: definition, identification and bounds},
  author={VanderWeele, Tyler J},
  journal={Scandinavian Journal of Statistics},
  volume={38},
  number={3},
  pages={551--563},
  year={2011},
  publisher={Wiley Online Library}
}

@article{vansteelandt2014structural,
  title={Structural nested models and G-estimation: the partially realized promise},
  author={Vansteelandt, Stijn and Joffe, Marshall and others},
  journal={Statistical Science},
  volume={29},
  number={4},
  pages={707--731},
  year={2014},
  publisher={Institute of Mathematical Statistics}
}

@book{wooldridge2010econometric,
  title={Econometric {A}nalysis of {C}ross {S}ection and {P}anel {D}ata},
  author={Wooldridge, Jeffrey M},
  year={2010},
  publisher={MIT press}
}

@article{ying2019two,
  title={Two-stage residual inclusion for survival data and competing risks—An instrumental variable approach with application to SEER-Medicare linked data},
  author={Ying, Andrew and Xu, Ronghui and Murphy, James},
  journal={Statistics in Medicine},
  volume={38},
  number={10},
  pages={1775--1801},
  year={2019},
  publisher={Wiley Online Library}
}

@Manual{ying2020ivsacim,
    title = {ivsacim: Structural Additive Cumulative Intensity Models with IV},
    author = {Andrew Ying},
    year = {2022},
    note = {R package version 2.0.0},
    url = {https://CRAN.R-project.org/package=ivsacim}
}

@article{ying2022structural,
  title={Structural cumulative survival models for estimation of treatment effects accounting for treatment switching in randomized experiments},
  author={Ying, Andrew and Tchetgen Tchetgen, Eric J.},
  journal={Biometrics},
  year = {2022},
  publisher={Wiley Online Library}
}

\newpage
\appendix

\section{Identification}
Indeed, we can show
\begin{prp}\label{prp:iden}
Under model \eqref{eq:condstructuralinvalid} and Assumptions \ref{assump:ivrelevance}- \ref{assump:condindcensoring}, the parameter of interest $(B_D(t), B_Z(t))^\top$ solves the following two sequences of population-level estimating equations at each time $t$,
\begin{align}
    \E\left[
    \begin{pmatrix}
        Z^c\\
        Z^cD(t)^c
    \end{pmatrix}
    e^{\int_0^{t-}D(s)dB_D(s) + ZdB_Z(s)}Y(t)\left\{dN(t) - D(t)dB_D(t) - ZdB_Z(t)\right\}\right] = 0,\label{eq:unidenee}
\end{align}
where $Z^c = Z - \E(Z|L)$ and $D(t)^c = D(t) - \E\{D(t)|Z, L\}$. 
\end{prp}
However, the above Equation \eqref{eq:unidenee} cannot identify $(B_D(t), B_Z(t))^\top$ because the design matrix
\begin{align}
    &\bbM(t)\\
    &= \E \left[
    \begin{pmatrix}
    Z^cD(t) & Z^cD(t)^cD(t)\\
    Z^cZ &Z^cD(t)^cZ
    \end{pmatrix}
    Y(s)\exp\left\{\int_0^{t-}D(u)dB_D(u) + ZdB_Z(u)\right\}\right].\label{eq:identification}
\end{align}
might be degenerate at some time $t$. Especially, consider the case when there is indeed no treatment switching, then \eqref{eq:identification} becomes
\begin{align}
    \bbM(t)= \E \left[
    \begin{pmatrix}
    Z^cZ & Z^cZ^cZ\\
    Z^cZ &Z^cZ^cZ
    \end{pmatrix}
    Y(s)\exp\left\{\int_0^{t-}D(u)dB_D(u) + ZdB_Z(u)\right\}\right],
\end{align}
which is degenerate. Therefore intuitively the uniqueness of the solutions to \eqref{eq:identification} necessitates enough rates of treatment switching across $t \in [0, \tau]$. This kind of assumption can be awkward in practice. Therefore to avoid such an assumption, one may leverage the Moore-Penrose generalized inverse that is commonly adopted in survival analysis literature \citep{martinussen2006dynamic} and impose the following weaker assumption
\begin{assump}\label{assump:uniqsolu}
$(B_D(t), B_Z(t))^\top$ is the continuous solution to the equation $B(t) = \Upsilon(B, t)$ with minimum $L_2$ norm at each $t$.
\end{assump}
The validity of this regularity assumption is beyond the scope of this statistical paper but can be of interest for researchers on differential equations. With this and Proposition \ref{prp:iden}, one can show that $(B_D(t), B_Z(t))^\top$ admits a closed form solution
\begin{equation}\label{eq:solution}
\begin{pmatrix}
B_D(t)\\
B_Z(t)
\end{pmatrix}
= \int_0^t\{\bbM(s)\}^\dagger \E \left[
\begin{pmatrix}
Z^c\\
Z^cD(s)^c
\end{pmatrix}\exp\left\{\int_0^{s-}D(u)dB_D(u) + ZdB_Z(u)\right\}dN(s)\right],
\end{equation}
where $(\cdot)^\dagger$ is the Moore-Penrose generalized inverse.

Note that since \eqref{eq:const} is a submodel of \eqref{eq:condstructuralinvalid}, this result carries over to that for $(\beta_D, \beta_Z)^\top$ under \eqref{eq:condstructuralinvalid}.

\section{Proofs}

\subsection{Proof of Proposition \ref{prp:iden}}
We write $\tilde N(t) = \mathbbm{1}(T \le t)$ the counting process of event time. It is straightforward that \eqref{eq:condstructuralinvalid} implies
\begin{align}
    &\E(d\tilde N_{\bar D(m - 1), 0, z = 0}(t)|\bar D(m), Z, T_{\bar D(m - 1), 0, z = 0} \geq t)\\
    &= \E(d\tilde N_{\bar D(m), 0, Z}(t)|\bar D(m), Z, T_{\bar D(m), 0} \geq t) \\
    &~- D(m)\mathbbm{1}(m \leq t < m + 1) dB_D(t) - Z\mathbbm{1}(m \leq t < m + 1) dB_Z(t) \label{eq:countingdiff}.
\end{align}
For any time $t$ satisfying $m \le t < m + 1$, under \eqref{eq:countingdiff} and Assumption \ref{assump:condindcensoring},
\begin{eqnarray}
&&\E\Bigg[\begin{pmatrix}
Z^c\\
Z^cD(s)^c
\end{pmatrix}
e^{\int_0^{t-} D(s)dB_D(s) + ZdB_Z(s)}Y(t)\Big\{dN(t) - D(m)dB_D(t) - ZdB_Z(t)\Big\}\Bigg]\\
&=&\E\Bigg[\begin{pmatrix}
Z^c\\
Z^cD(s)^c
\end{pmatrix}
e^{\int_0^{t-} D(s)dB_D(s) + ZdB_Z(s)}Y(t)\Big\{d\tilde N_{\bar D(m), 0}(t) - D(m)dB_D(t) - ZdB_Z(t)\Big\}\Bigg]\label{eq:populationestieq}\\
&=&\E\Bigg[\begin{pmatrix}
Z^c\\
Z^cD(s)^c
\end{pmatrix}
e^{\int_0^{t-} D(s)dB_D(s) + ZdB_Z(s)}Y(t)\Big\{\E(d\tilde N_{\bar D(m), 0}(t)|\bar D(m), Z, \tilde T_{\bar D(m), 0} \geq t)\\
&& - D(m)dB_D(t) - ZdB_Z(t)\Big\}\Bigg]\\
&=&\E\Bigg[\begin{pmatrix}
Z^c\\
Z^cD(s)^c
\end{pmatrix}
e^{\int_0^{m-} D(s)dB_D(s) + ZdB_Z(s)}\frac{\mathbbm{1}(\tilde T(\bar D(m), 0) \ge t)}{e^{-\int_{m}^{t-} D(s)dB_D(s) + ZdB_Z(s)}}\mathbbm{1}(C \ge t)\\
&&~~~\cdot\Big[\E(d\tilde N_{\bar D(m), 0}(t)|\bar D(m), Z, \tilde T_{\bar D(m), 0} \geq t) - D(m)dB_D(t)- ZdB_Z(t)\Big]\Bigg]\\
&=&\E\Bigg[\begin{pmatrix}
Z^c\\
Z^cD(s)^c
\end{pmatrix}
e^{\int_0^{m-} D(s)dB_D(s) + ZdB_Z(s)}\\
&&\cdot\P(\tilde T(\bar D(m - 1), 0) \ge t|\bar D(m), Z, \tilde T_{\bar D(m), 0} \geq t)\mathbbm{1}(C \ge t)\\
&&~~~\cdot\Big[\E(d\tilde N_{\bar D(m), 0}(t)|\bar D(m), Z, \tilde T_{\bar D(m), 0} \geq t) - D(m)dB_D(t)- ZdB_Z(t)\Big]\Bigg]\\
&=&\E\Bigg[\begin{pmatrix}
Z^c\\
Z^cD(s)^c
\end{pmatrix}e^{\int_0^{m-} D(s)dB_D(s) + ZdB_Z(s)}Y_{\bar D(m - 1), 0, z= 0}(t)\mathbbm{1}(C \ge t)\\
&&~~~\cdot\Big[\E(d\tilde N_{\bar D(m), 0}(t)|\bar D(m), Z, \tilde T_{\bar D(m), 0} \geq t) - D(m)dB_D(t)- ZdB_Z(t)\Big]\Bigg]\\
&=&\E\Bigg[\begin{pmatrix}
Z^c\\
Z^cD(s)^c
\end{pmatrix}e^{\int_0^{m-} D(s)dB_D(s) + ZdB_Z(s)}Y_{\bar D(m - 1), 0, z= 0}(t)\mathbbm{1}(C \ge t)\\
&&~~~\cdot\Big[\E(d\tilde N_{\bar D(m - 1), 0, z= 0}(t)|\bar D(m), Z, \tilde T_{\bar D(m - 1), 0, z= 0} \geq t)]\Bigg]\\
&=&\E\Bigg[\begin{pmatrix}
Z^c\\
Z^cD(s)^c
\end{pmatrix}e^{\int_0^{m-} D(s)dB_D(s) + ZdB_Z(s)}Y_{\bar D(m - 1), 0, z= 0}(t)\mathbbm{1}(C \ge t)d\tilde N_{\bar D(m - 1), 0, z= 0}(t)\Bigg].
\end{eqnarray}
Now we can repeat these steps in order to blip down the effect of treatment $D(m - 1)$ at $m - 1$. However, note that for $m - 1 < t \leq m$, this effect is null by assumption, and $\int_0^{m-}D(s)dB_D(s)$ is a function of $\bar D(m - 1)$. It follows that
\begin{eqnarray}
&&\E\Bigg[\begin{pmatrix}
Z^c\\
Z^cD(s)^c
\end{pmatrix}e^{\int_0^{m-} D(s)dB_D(s) + ZdB_Z(s)}Y_{\bar D(m - 1), 0, z = 0}(t)\mathbbm{1}(C \ge t)d\tilde N_{\bar D(m - 1), 0, z= 0}(t)\Bigg]\\
&&\E\Bigg[\begin{pmatrix}
Z^c\\
Z^cD(s)^c
\end{pmatrix}e^{\int_0^{m-} D(s)dB_D(s) + ZdB_Z(s)}Y_{\bar D(m - 1), 0, z = 0}(t)\mathbbm{1}(C \ge t)\\
&&\cdot\E(d\tilde N_{\bar D(m - 1), 0, z= 0}(t)|\bar D(m - 1), Z, \tilde T_{\bar D(m - 1), 0, z = 0} \geq t)\Bigg]\\
&&\E\Bigg[\begin{pmatrix}
Z^c\\
Z^cD(s)^c
\end{pmatrix}e^{\int_0^{(m - 1)-} D(s)dB_D(s) + ZdB_Z(s)}\\
&&\cdot\frac{\P(T_{\bar D(m - 1), 0, z= 0} \geq t|\bar D(m - 1), Z, \tilde T_{\bar D(m - 1), 0, z = 0} \geq t)}{e^{-\int_{m - 1}^{t\wedge (m-)}D(s)dB_D(s)}}\mathbbm{1}(C \ge t)\\
&&~~\cdot\E(d\tilde N_{\bar D(m - 1), 0, z= 0}(t)|\bar D(m - 1), Z, \tilde T_{\bar D(m - 1), 0, z = 0} \geq t)\Bigg]\\
&&\E\Bigg[\begin{pmatrix}
Z^c\\
Z^cD(s)^c
\end{pmatrix}e^{\int_0^{(m - 1)-} D(s)dB_D(s) + ZdB_Z(s)}\\
&&\P(T_{\bar D(t_{m - 2}), 0} \geq t|\bar D(m - 1), Z, \tilde T_{\bar D(m - 1), 0, z = 0} \geq t)\mathbbm{1}(C \ge t)\\
&&~~\cdot\E(d\tilde N_{\bar D(m - 1), 0, z= 0}(t)|\bar D(m - 1), Z, \tilde T_{\bar D(m - 1), 0, z = 0} \geq t)\Bigg]\\
&=&\E\Bigg[\begin{pmatrix}
Z^c\\
Z^cD(s)^c
\end{pmatrix}e^{\int_0^{(m - 1)-} D(s)dB_D(s) + ZdB_Z(s)}Y_{\bar D(t_{m - 2}), 0, z = 0}(t)\mathbbm{1}(C \ge t)\\
&&\E(d\tilde N_{\bar D(m - 1), 0, z= 0}(t)|\bar D(m - 1), Z, \tilde T_{\bar D(m - 1), 0, z = 0} \geq t)\Bigg]\\
&=&\E\Bigg[\begin{pmatrix}
Z^c\\
Z^cD(s)^c
\end{pmatrix}e^{\int_0^{(m - 1)-} D(s)dB_D(s) + ZdB_Z(s)}Y_{\bar D(t_{m - 2}), 0, z = 0}(t)\mathbbm{1}(C \ge t)\\
&&\E(d\tilde N_{\bar D(m - 2), 0, z = 0}(t)|\bar D(m - 1), Z, \tilde T_{\bar D(t_{m - 2}), 0, z = 0} \geq t)\Bigg]\\
&=&\E\Bigg[\begin{pmatrix}
Z^c\\
Z^cD(s)^c
\end{pmatrix}e^{\int_0^{(m - 1)-} D(s)dB_D(s) + ZdB_Z(s)}Y_{\bar D(t_{m - 2}), 0, z = 0}(t)\mathbbm{1}(C \ge t)d\tilde N_{\bar D(t_{m - 2}), 0}(t)\Bigg].
\end{eqnarray} 
A recursive application of the above argument yields
\begin{eqnarray}
&&\E\Bigg[\begin{pmatrix}
Z^c\\
Z^cD(s)^c
\end{pmatrix}\exp\Bigg(\int_0^{t-} D(s)dB_D(s)\Bigg)\Big[dN_{\bar D(m), 0}(t) - Y(t)D(m)dB_D(t)\Big]\Bigg]\\
&=&\E\Bigg[\begin{pmatrix}
Z^c\\
Z^cD(s)^c
\end{pmatrix}Y_{0, z= 0}(t)\mathbbm{1}(C \ge t)dN_{0}(t)\Bigg]\\
&=&\E\Bigg[\begin{pmatrix}
Z^c\\
Z^cD(s)^c
\end{pmatrix}Y_{0, z= 0}(t)dN_{0, z= 0}(t)\P(C \ge t|L)\Bigg]\\
&=&\E\Bigg[\begin{pmatrix}
Z^cY_0(t)dN_{0, z = 0}(t)\P(C \ge t|L)\\
Z^c\alpha_2(t; \tilde T(0, z = 0))Y_0(t)dN_{0, z = 0}(t)\P(C \ge t|L)
\end{pmatrix}\Bigg]= 0,
\end{eqnarray}
by the IV independence assumption. Consequently, \eqref{eq:identification} holds.
\qed

We introduce additional notation, some technical assumptions, and a necessary Helly's selection theorem to prove Theorem \ref{thm:cons}. For the proofs below, we plug in the truth $\E(Z|L)$ and $\E(D(t)|Z, L)$ into the estimators. We prove their uniform consistency and asymptotic normality. For two-step estimators given in the main text with regular and asymptotically linear estimates $\hat \E(Z|L)$ and $\hat \E(D(t)|Z, L)$, a Taylor expansion can be added into the proofs.

To facilitate exposition, we write $B(t) = (B_D(t), B_Z(t))^\top$. Define 
\begin{equation}
\Upsilon(B, t) = \int_0^t [\bbM\{B, s)\}]^\dagger\E\left\{ 
\begin{pmatrix}
Z^c\\
Z^cD(s)^c
\end{pmatrix}e^{\int_0^{s-}(D(u), Z)dB(u)}dN(s)\right\}.
\end{equation}
Note that $B(t)$ is the solution to $B(t) = \Upsilon(B, t)$ by results in the last section.

We write $\lVert g \rVert_\infty = \sup_{t \in [0, \tau]} |g(t)|$ and use $\cV(g)$ to denote the total variation of $g$ over the interval $[0, \tau]$. Let $M^\circ = \lVert B(t) \rVert_\infty < \infty$. Define 
\begin{equation}
    \hat \bbM(B, s) = \E_n \bigg\{
    \begin{pmatrix}
    Z^cD(s) & Z^cD(s)^cD(s)\\
    Z^cZ &Z^cD(s)^cZ
    \end{pmatrix}
    Y(s)e^{\int_0^{s-}\{D(u), Z\}dB(u)}\bigg\},
\end{equation}
\begin{equation}
    \Upsilon_n(B, t) = \int_0^t \{\hat \bbM(B, s)\}^\dagger \E_n\left\{\begin{pmatrix}
Z^c\\
Z^cD(s)^c
\end{pmatrix}e^{\int_0^{s-}\{D(u), Z\}dB(u)}dN(s)\right\},
\end{equation}
and
\begin{equation}
\bar \Upsilon_n(B, t) = \int_0^t \{\bbM(B, s)\}^\dagger \E_n\left\{
\begin{pmatrix}
Z^c\\
Z^cD(s)^c
\end{pmatrix}e^{\int_0^{s-}\{D(u), Z\}dB(u)}dN(s)\right\}.
\end{equation}
We note for later reference that for any two functions $B_1$ and $B_2$, 
\begin{equation}\label{eq:gammaupperdiff}
    \lVert \Upsilon(\xi(B_1), t) - \Upsilon(\xi(B_2), t)\rVert_\infty \le 4\tau M_{\text{max}}\exp(M)/\nu \lVert B_1 - B_2\rVert_\infty.
\end{equation}


The following regularity condition is trivially satisfied for binary $D(t)$ and $Z$. 
\begin{assump}\label{assump:ivbound}
The instrument $Z$ and the treatment process $D(t)$ is uniformly bounded by $M_{\text{max}}$.
\end{assump}
Define
\begin{equation}
    \bbM(B(\cdot), t) = \E \left\{
    \begin{pmatrix}
    Z^cD(t) & Z^cD(t)^cD(t)\\
    Z^cZ &Z^cD(t)^cZ
    \end{pmatrix}
    Y(t)e^{\int_0^{t-}(D(s), Z)dB(s)}D(t)\right\},
\end{equation}
for any $B \in \bbB$, where $\bbB$ is the set of two-dimensional functions on $[0, \tau]$ that have total variations bounded by some $M > M^\circ$. 

We shall use Helly's selection theorem to establish Theorem \ref{thm:cons}.
\begin{thm}[Helly's Selection Theorem]
Let $\{f_n\}$ be a sequence of functions on $[0, \tau]$ such that $\lVert f_n \rVert_\infty \le  A_1$ and $\sup_n\cV(f_n) \le A_2$, where $A_1$ and $A_2$ are finite constants. Then
\begin{enumerate}
    \item There exists a subsequence $\{f_{n_j}\}$ of $\{f_n\}$ which converges pointwise to some function $f$.
    \item If $f$ is continuous, the convergence is uniform.
\end{enumerate}
\end{thm}

\begin{lem}\label{lem:Aconv}
\begin{equation}
    \sup_{s \in [0, \tau], B \in \bbB}|\hat \bbM(B, s) - \bbM(B, s)| \to 0~~\text{a.s.}.
\end{equation}
\end{lem}
\begin{proof}
Define
\begin{equation}
    \phi_{B, s}\{T, \Delta, Z, \bar D\} := \begin{pmatrix}
    Z^cD(s) & Z^cD(s)^cD(s)\\
    Z^cZ &Z^cD(s)^cZ
    \end{pmatrix}Y(s)e^{\int_0^{s-}(D(u), Z)dB(u)}.
\end{equation}
Hence $\hat \bbM(B, s) - \bbM(B, s) = (\P_n - \P)\phi_{B, s}$. To prove the Lemma it suffices to show that $\{\phi_{B, s}: B \in \bbB, s \in [0, \tau]\}$ is Gilvenko-Cantelli \citep{van1996weak}. Note that all functions that are of bounded variation form a Gilvenko-Cantelli class \cite[Example 19.11]{van2000asymptotic}. This is immediate since $\max_{H, s}\cV(\phi_{B, s})$ can be proved to be finite easily.
\end{proof}
\subsection{Proof of Theorem \ref{thm:cons}}\label{sec:proofcons}

We give a roadmap for the proof of consistency.
\begin{enumerate}
    \item We construct a modified version $\tilde B(t) = \tilde B_{n}(t)$ of the estimator $\hat B(t)$ that is uniformly of bounded variation over $n$.
    \item We show that 
    \begin{equation}\label{eq:gammaconverge}
        \sup_{s \in [0, \tau], B \in \bbB}|\Upsilon_n(B, s) - \Upsilon(B, s)| \to 0~~ \text{a.s.}.
    \end{equation}
    \item These, together with the Helly's Selection Theorem, imply $\lVert\tilde B(t) - B(t)\rVert_\infty \to 0$ a.s..
    \item Finally we show that $\tilde B(t)$ is equal to $\hat B(t)$ in large samples, and therefore, $\|\hat B(t) - B(t)\|_\infty \to 0,$ a.s..
\end{enumerate}

\bigskip
\noindent
{\sc Step 1:} Let $\xi(y) = \text{sgn}(y) \min(|y|, M)$. We define the modified estimator $\tilde B_n$ to be the solution to the equation $B(t) = \Upsilon_n\{\xi(B), t\}$. Note that by Lemma \ref{lem:Aconv} and Assumption \ref{assump:uniqsolu}, the total variation
\begin{equation}
    \cV(\tilde B) \le 4\tau M_{\text{max}}\exp(M) /\nu
\end{equation}
Therefore $\tilde B(t)$ is uniformly of bounded variation over $n$. Also, since $\tilde B_D(0) = 0$, $\tilde B_Z(0) = 0$, $\tilde B$ is uniformly bounded, therefore Helly's selection theorem applies. Note that $\Upsilon\{\xi(B), t\} = \Upsilon(B, t) = B$.

\bigskip
\noindent
{\sc Step 2:} We want to show
\begin{equation}\label{eq:barupsilonconverge}
    \sup_{s \in [0, \tau], B \in \bbB}|\bar \Upsilon_n(B, s) - \Upsilon(B, s)| \to 0, ~~\text{a.s.}.
\end{equation}
To this end, we first define
\begin{equation}
    \psi_{H,t}(T, \Delta, Z, D(\cdot)) := (\bbM(H(\cdot), t))^\dagger \begin{pmatrix}
Z^c\\
Z^cD(s)^c
\end{pmatrix}e^{\int_0^{T-}(D(s), Z - D(s))dH(s)}N(t),
\end{equation}
and therefore $\bar \Upsilon_n(H, t) = \E_n\psi_{H,t}$ and $\Upsilon(H, t) = \P\psi_{H,t}$. Then it suffices to show that $\{\psi_{H,t}: B \in \bbB, t \in [0, \tau]\}$ is Gilvenko-Cantelli. This result is an immediate consequence of the following facts:
\begin{enumerate}
    \item All functions that are of bounded variation form a Gilvenko-Cantelli class \cite[Example 19.11]{van2000asymptotic}. Therefore the function class $\{(Z^c, Z^cD(s)^c)^\top e^{\int_0^{T-}(D(s), Z - D(s))dH(s)}N(t): B \in \bbB,~t \in [0, \tau]\}$ is Gilvenko-Cantelli.
    \item A Gilvenko-Cantelli class multiplied by a uniformly bounded, measurable function remains Gilvenko-Cantelli \cite[Example 2.10.10]{van1996weak}.
\end{enumerate}

\begin{equation}
    \sup_{s \in [0, \tau], B \in \bbB}|\bar \Upsilon_n(B, s) - \Upsilon(B, s)| = \|(\P_n - \P)\psi_{H,t}\|_\infty \to 0, ~~\text{a.s.}.
\end{equation}
Now by Lemma \ref{lem:Aconv},
\begin{equation}
    \sup_{s \in [0, \tau], B \in \bbB}|\bar \Upsilon_n(B, s) - \Upsilon_n(B, s)| \to 0, ~~\text{a.s.}.
\end{equation}
These two imply \eqref{eq:gammaconverge}.

\bigskip
\noindent
{\sc Step 3:} We prove $\lVert\tilde B(t) - B(t)\rVert_\infty \to 0$ a.s., by contradiction. Without loss of generality, we assume
\begin{equation}
    \liminf_{n \to \infty}\lVert\tilde B(t) - B(t)\rVert_\infty > 0.
\end{equation}
By Helly's Selection Theorem, there exists a subsequence $\{n_j\}$ such that $\tilde B_{n_j}(t) - B(t)$ converges to some limit $H(t)$. We further claim this limit is continuous, in fact, Lipschitz continuous. To see this, for any $t_1 < t_2$ in $[0, \tau]$, when $n_j$ is large enough
\begin{align}
|B(t_2) - B(t_1)| &\le |B(t_2) - \tilde B_{n_j}(t_2)| + |\tilde B_{n_j}(t_2) - \tilde B_{n_j}(t_1)| \\
&~~+ |\tilde B_{n_j}(t_1) - B(t_1)|\\
&\le \bigg|\int_{t_1}^{t_2}(\hat \bbM(s))^\dagger\E_{n_j} \begin{pmatrix}
Z^c\\
Z^cD(s)^c
\end{pmatrix}e^{\int_0^{s-}(D(u), Z)d\xi\{\tilde B_{n_j}(u)\}}dN(s)\bigg| + 2\eps\\
&\le 4 M_{\text{max}} \exp(M)/\nu(t_2 - t_1)+ 2\eps,
\end{align}
for any $\eps > 0$. Therefore, by the second part of Helly’s theorem, the convergence of the sub-subsequence is uniform. Going further, the limit $B$ satisfies $B(t) = \Upsilon(B, t)$ since
\begin{align*}
\lVert B(t) -  \Upsilon(\xi(B), t)\rVert_\infty &\le \lVert B(t) - \tilde B_{n_j}(t)\rVert_\infty + \lVert \tilde B_{n_j}(t) - \Upsilon_{n_j}\{\xi(\tilde B_{n_j}), t\}\rVert_\infty \\
&~~+ \lVert \Upsilon_{n_j}\{\xi(\tilde B_{n_j}), t\} -  \Upsilon(B, t)\rVert_\infty + \lVert \Upsilon\{\xi(\tilde B_{n_j}), t\} -  \Upsilon(B, t)\rVert_\infty\\
&~~\to 0,
\end{align*}
where the last statement results from the uniform convergence of $\tilde B_{n_j}(t)$, the definition of $\tilde B_{n_j}(t)$, \eqref{eq:gammaconverge} and \eqref{eq:gammaupperdiff}. The fact that the solution $B$ to $B(t) = \Upsilon(B, t)$ is unique by Assumption \ref{assump:uniqsolu} leads to a contradiction. It hence follows that $\lVert\tilde B(t) - B(t)\rVert_\infty \to 0$ a.s..

\bigskip
\noindent
{\sc Step 4:} Since $\lVert B(t) \rVert_\infty \le M^\circ$ and we just showed $\lVert \tilde B(t) - B(t)\rVert_\infty \to 0$ a.s., for sufficiently large $n$ we have $\lVert \tilde B(t) \rVert_\infty \le M^\circ + \frac{1}{2}(M - M^\circ) < M$, and therefore $\xi(\tilde B(t)) = \tilde B(t)$. Therefore, for sufficiently large $n$, $\tilde B(t)$ satisfies $\tilde B(t) = \Upsilon\{\xi(\tilde B), t\} =  \Upsilon(\tilde B, t)$, or in other words, $\hat B(t) = \tilde B(t)$. We thus have $\sup_{t \in [0, \tau]} |\hat B(t) - B(t)| \to 0,$ a.s.

The consistency of $(\hat B_D(t, \hat \theta), \hat B_Z(t, \hat \theta))^\top$ then follows immediately by a Taylor series expansion since $\hat \theta$ is consistent for $\theta$.

For asymptotic normality, we aim to provide a sum of i.i.d. representation of $\sqrt{n}\{\hat B(t, \theta) - B(t, \theta)\}$ heuristically. To that end, we rewrite the normalized residual 
\begin{equation}
    V_n(t, \theta) := \sqrt{n}\{\hat B(t, \theta) - B(t, \theta)\}
\end{equation} 
at a fixed $\theta$ as a solution to a Volterra equation, which shall yield a sum of i.i.d. representation at each time point $t$ in this case. To establish this, we integrate by part the Riemann–Stieltjes integral (the integral is interpreted pathwise for the C\`adl\`ag stochastic process $D_i(s)$, $B_D$ and $B_Z$ are assumed to be continuous),
\begin{align}
\int_0^{s-} (D_i(u), Z_i) dB(u) &= (D_i(s), Z_i) B(s-) - \int_0^{s-} B(u) d(D_i(u), Z_i)^\top \\
&= (D_i(s), Z_i) B(s-) - \int_0^{s-} B_D(u) dD_i(u) \\
&=: G_{1, i}(s-) + G_{2, i}(s-). \label{eq:integralbyparts}
\end{align}
By Theorem \ref{thm:cons} and an application of Slutsky's Theorem, we can write 
\begin{equation}
    V_n(t, \theta) = \frac{1}{\sqrt{n}}\int_0^t \sum_{i = 1}^n H_i(s, \hat B) \{dN_i(s) - (D_i(s), Z_i)dB(s, \theta)\} + o_p(1), 
\end{equation}
where 
\begin{equation}
    H_i(s, B) := \{M(s)\}^\dagger (Z_i^c, Z_iD_i^c(s))^\top \exp\left\{\int_0^{s-} (D_i(u), Z) d\hat B(u, \theta)\right\}.
\end{equation}

Now we are ready to write down the Volterra equation. By consistency and a Taylor expansion, it is easy to see that
\begin{eqnarray}
V_n(t, \theta) &=& \frac{1}{\sqrt{n}}\int_0^t \sum_{i = 1}^n H_i(s, B)\{dN_i(s) - (D_i(s), Z_i)dB(s, \theta)\}\\
&&~~~+\int_0^t V_n(s-, \theta)\{1 + o_p(1)\}\sum_{i = 1}^n\frac{\partial H_i(s, B)}{\partial B(s-, \theta)}dN_i(s) + o_p(1), 
\end{eqnarray}
which by \eqref{eq:integralbyparts} yields
\begin{eqnarray}
V_n(t, \theta) &=& \frac{1}{\sqrt{n}}\int_0^t \sum_{i = 1}^n H_i(s, B)\{dN_i(s) - (D_i(s), Z_i)dB(s, \theta)\}\\
&&+\int_0^t V_n(s-, \theta)\{1 + o_p(1)\} \sum_{i = 1}^n\frac{\partial H_i(s, B)}{\partial G_{1, i}(s-)}(D_i(s-), Z_i) dN_i(s)\\
&&+\int_0^t\sum_{i = 1}^n\frac{\partial H_i(s, B)}{\partial G_{2, i}(s-)}\left[\int_0^{s-}V_n(u, \theta)\{1 + o_p(1)\}d(D_i(u), Z_i)\right] dN_i(s)\\
&&~~+ o_p(1)\\
&=& \frac{1}{\sqrt{n}}\int_0^t \sum_{i = 1}^n H_i(s, B)\{dN_i(s) - (D_i(s), Z_i)dB(s, \theta)\}\\
&&+\int_0^t V_n(s-, \theta)\{1 + o_p(1)\} \sum_{i = 1}^n\frac{\partial H_i(s, B)}{\partial G_{1, i}(s-)}(D_i(s-), Z_i) dN_i(s)\\
&&+ \int_0^t V_n(s, \theta)\sum_{i = 1}^n \left\{\int_{s+}^t \frac{\partial H_i(u, B_D)}{\partial G_{2, i}(u-)} dN_i(u)\right\} (dD_i(s), dZ_i),
\end{eqnarray}
where the last equation follows by Fubini's theorem. Together we have the Volterra-equation,
\begin{align}
V_n(t, \theta) &= \frac{1}{\sqrt{n}}\int_0^t \sum_{i = 1}^n H_i(s, B)\{dN_i(s) - (D_i(s), Z_i)dB(s, \theta)\}\\
&+\int_0^t V_n(s-, \theta) \sum_{i = 1}^n\left[\frac{\partial H_i(s, B)}{\partial G_{1, i}(s-)} (D_i(s-), Z_i)dN_i(s) + \left\{\int_{s+}^t \frac{\partial H_i(u, B_D)}{\partial G_{2, i}(u-)} dN_i(u)\right\}(dD_i(s), dZ_i)\right],\label{eq:volterra}
\end{align}
which admits a solution with explicit form given by
\begin{equation}
    V_n(t, \theta) = \frac{1}{\sqrt{n}}\int_0^t \cF(s, t)\sum_{i = 1}^n H_i(s, B)\{dN_i(s) - (D_i(s), Z_i)dB(s, \theta)\} + o_p(1),
\end{equation}
where
\begin{equation}
    \cF(s, t) := \prod_{(s, t]} \left(1 + \sum_{i = 1}^n\left[\frac{\partial H_i(u, B)}{\partial G_{1, i}(u-)}(D_i(u), Z_i) dN_i(u) + \left\{\int_u^t\frac{\partial H_i(\cdot, B)}{\partial G_{2, i}(\cdot-)} dN_i(\cdot)\right\}d(D_i(u), Z_i)\right]\right). 
\end{equation}
This leads to an i.i.d. representation
\begin{equation}
   V_n(t, \theta) = \frac{1}{\sqrt{n}}\sum_{i = 1}^n \epsilon_i^B(t)+ o_p(1), 
\end{equation}
with the $\epsilon_i^B(t)$ being zero-mean i.i.d. terms, defined as
\begin{equation}
    \epsilon_i^B(t) := \int_0^t \cF(s, t) H_i(s, B)\{dN_i(s) - D_i(s)dB_D(s)\}.
\end{equation}
This, together with a Taylor expansion, gives
\begin{eqnarray}
&&\sqrt{n}\{\hat B(t, \hat \theta) - B(t, \theta)\} \\
&=& \sqrt{n}\{\hat B(t, \theta) - B(t, \theta)\} + \sqrt{n}\{\hat B(t, \hat \theta) - B(t, \theta)\}\\
&=& \sqrt{n}\{\hat B(t, \theta) - B(t, \theta)\} + \frac{\partial B(t, \theta)}{\partial \theta}\bigg|_{\hat \theta}\{1 + o_p(1)\}\sqrt{n}(\hat \theta - \theta)+ o_p(1).
\end{eqnarray}
Finally we have 
\begin{equation}
    \sqrt{n}\{\hat B(t, \hat \theta) - B(t, \theta)\} = \frac{1}{\sqrt{n}}\sum_{i = 1}^n \epsilon_i(t, \theta) + o_p(1),
\end{equation}
where
\begin{equation}\label{eq:empiricalerror}
    \epsilon_i(t, \theta) := \epsilon_i^B(t) + \frac{\partial B_Z(t, \theta)}{\partial \theta}\bigg|_{\theta}\epsilon_i^\theta.
\end{equation}
With this i.i.d. representation \eqref{eq:empiricalerror}, similar arguments as in \citet{lin2000semiparametric, martinussen2017instrumental} can be adopted to establish the convergence of $V_n(t)$ in distribution to a Gaussian process. Also, \eqref{eq:empiricalerror} easily implies a variance estimate.

\section{Simulation Compatibility}
We establish that our proposed data generating process in the simulation study ensures that Asssumption \ref{assump:ivnointeraction} holds and is compatible with our model \eqref{eq:condstructuralinvalid}, where $B_D(t) = 0.2 t$ and $B_Z(t) = 0.1 t$.

To show that Asssumption \ref{assump:ivnointeraction} holds, note that we first generate $U$ from a normal distribution and $Z$ from a random binomial distribution. Then we generate a shift variable $W$ such that
\begin{equation}
    \P(W > t|Z, U) = \alpha(t; U)(2Z - 1) + \beta_0(t) + Z\beta_1(t),
\end{equation}
where we should pick $\alpha()$ and $\beta()$ so $\P(W > t|Z, U)$ is decreasing in $t$ and between $[0, 1]$. To make sure that the RHS indeed gives a valid distribution function for any $Z$ and $U$, we need 
\begin{equation}
    \beta_0(0) = 1, ~~~\alpha(0; U) = 0, ~~~\beta_1(0) = 0,
\end{equation}
and $\alpha(t; U) + \beta_0(t) + \beta_1(t)$ and $-\alpha(t; U) + \beta_0(t)$ nonnegative and decreasing. If choosing $\beta_0(t)$ as a regular distribution function, then $\alpha(t; U)$ and $\beta_1(t)$ have to converge to $0$ as $t \to \infty$.  Because $D(t) = \mathbbm{1}(W > t)Z + \mathbbm{1}(W \leq t)(1 - Z)$, we have
\begin{align}
    \E(D(t)|Z, U) &= Z\P(W > t|Z, U) + (1 - Z)(1 - \P(W > t|Z, U))\\
    &=(2Z - 1)\P(W > t|Z, U) + 1 - Z\\
    &=(2Z - 1)\beta(t; Z) + 1 - Z + \alpha(t; U).
\end{align}
Therefore $\alpha_Z(t; Z, L) = (2Z - 1)(\beta_0(t) + Z\beta_1(t)) + 1 - Z$ and $\alpha_U(t; U) = \alpha(t; U)$.

To show that our data generating process is compatible with our model \eqref{eq:condstructuralinvalid}, where $B_D(t) = 0.2 t$ and $B_Z(t) = 0.1 t$, it suffices to show that,
\begin{align}
    &\frac{\P\left(\tilde T(\bar d(m), \vec 0, Z) > t|\bar D(m) = \bar d(m), Z, U, \tilde T \ge m\right)}{\P\left(\tilde T(\bar d(m - 1), \vec 0, z = 0) > t|\bar D(m) = \bar d(m), Z = 0, U, \tilde T \ge m\right) } \\
    &= \exp\left\{-0.2\int_{m}^{t \wedge m + 1} d(m) dt - 0.1\int_{m}^{t \wedge m + 1} Z dt\right\}, \label{eq:simucondU}
\end{align}
which after integrating $U$ out, yields \eqref{eq:condstructuralinvalid}. We work on the numerator of the LHS of \eqref{eq:simucondU} first,
\begin{align}
    &\P\left(\tilde T(\bar d(m), \vec 0, Z) > t|\bar D(m) = \bar d(m), Z, U, \tilde T \ge m\right)\\
    &=\frac{\P\left(\tilde T(\bar d(m), \vec 0, Z) > t|\bar D(m) = \bar d(m), Z, U\right)}{\P\left(\tilde T > m|\bar D(m) = \bar d(m), Z, U\right)}=\frac{\P\left(\tilde T(\bar d(m), 0, z) > t|U\right)}{\P\left(\tilde T(\bar d(m), 0, z) > m|U\right)}\\
    &= \exp\left(-0.1  (t - m) - 0.2 \int_{m}^td(s)ds - 0.1 \int_{m}^tZds  - 0.15  U_{2}  (t - m)\right),
\end{align}
where the second equation follows by consistency and unconfoundedness (conditional on U) guaranteed by our data generating process. Now we turn to the denominator of the LHS of \eqref{eq:simucondU},
\begin{align}
    &\P\left(\tilde T(\bar d(m - 1), \vec 0, z = 0) > t|\bar D(m) = \bar d(m), Z = 0, U, \tilde T \ge m\right)\\
    &=\frac{\P\left(\tilde T(\bar d(m - 1), \vec 0, Z = 0) > t|\bar D(m) = \bar d(m), Z = 0, U\right)}{\P\left(\tilde T \ge m|\bar D(m) = \bar d(m), Z = 0, U\right)} = \frac{\P\left(\tilde T(\bar d(m - 1), \vec 0, Z = 0) > t|U\right)}{\P\left(\tilde T(\bar d(m - 1), \vec 0, Z = 0) \ge m|U\right)}\\
    &= \exp\left(-0.1  (t - m)  - 0.15  U_{2}  (t - m)\right),
\end{align}
now \eqref{eq:simucondU} is straightforward.

\end{document}